\documentclass[a4paper, UKenglish, cleveref, autoref, thm-restate]{lipics-v2021}

\graphicspath{{./figures/}}

\title{\general{Querying in Constant Expected Time with Learned Indexes}}

\author{Luis Alberto Croquevielle}{Imperial College London}{a.croquevielle22@imperial.ac.uk}{https://orcid.org/0009-0002-0101-4431}{}

\author{Guang Yang}{Imperial College London}{guang.yang15@imperial.ac.uk}{https://orcid.org/0000-0001-6456-9077}{}

\author{Liang Liang}{Imperial College London}{liang.liang20@imperial.ac.uk}{https://orcid.org/0000-0002-4566-6178}{}
\authorrunning{Authors}

\author{Ali Hadian}{Imperial College London}{ali.hadian@gmail.com}{https://orcid.org/0000-0003-2010-0765}{}

\author{Thomas Heinis}{Imperial College London}{t.heinis@imperial.ac.uk}{https://orcid.org/0000-0002-7470-2123}{}



\category{}

\relatedversion{}

\nolinenumbers

\usepackage{mathtools}
\usepackage[ruled]{algorithm2e}
\usepackage{multirow}
\usepackage{makecell}
\usepackage{subcaption}
\usepackage{tikz}

\expandafter\def\expandafter\normalsize\expandafter{%
    \normalsize%
    \setlength\abovedisplayskip{4pt}%
    \setlength\belowdisplayskip{4pt}%
    \setlength\abovedisplayshortskip{-8pt}%
    \setlength\belowdisplayshortskip{2pt}%
}

\usetikzlibrary{decorations.pathreplacing}
\tikzset{
    JB tick/.style={
        draw=blue,
        fill=blue,
        rectangle,
        minimum width=1pt,
        minimum height=10mm,
        inner sep=0pt
    },
    J tick/.style={
        draw=blue,
        fill=blue,
        rectangle,
        minimum width=1pt,
        minimum height=5mm,
        inner sep=0pt
    },
    I tick/.style={
        draw=red,
        fill=red,
        rectangle,
        minimum width=1pt,
        minimum height=3mm,
        inner sep=0pt
    },
    IB tick/.style={
        draw=red,
        fill=red,
        rectangle,
        minimum width=1pt,
        minimum height=6mm,
        inner sep=0pt
    }
}

\definecolor{ForestGreen}{RGB}{34,139,34}

\newcommand{\general}[1]{{ #1}}
\newcommand{\probanalysis}[1]{{ #1}}
\newcommand{\reviewerone}[1]{{ #1}}
\newcommand{\reviewertwo}[1]{{ #1}}

\SetCommentSty{commentfont}

\DeclareRobustCommand{\svdots}{
  \vbox{%
    \baselineskip=0.3\normalbaselineskip
    \lineskiplimit=0pt
    \hbox{.}\hbox{.}\hbox{.}%
    \kern-0.9\baselineskip
  }%
}

\newtheorem{fact}{Fact}

\newcommand{\densityx}{f}
\newcommand{\densityq}{g}
\newcommand{\rhoh}{\rho}
\newcommand{\rhohf}{\rhoh_{\scriptscriptstyle \densityx}}
\newcommand{\rhohg}{\rhoh_{\densityq}}
\newcommand{\rhohat}{\hat{\rho}_\densityx}
\newcommand{\rank}{\mathbf{rank}}
\newcommand{\rankh}{\mathbf{\hat{r}}}
\newcommand{\indf}{\mathbf{1}}
\DeclarePairedDelimiter\norm{\lVert}{\rVert}
\DeclarePairedDelimiter{\ceil}{\lceil}{\rceil}
\newcommand{\osx}[1]{X_{(#1)}}
\newcommand{\prob}[1]{\mathbb{P}\left(#1\right)}
\newcommand{\expec}[1]{\mathbb{E}\left[#1\right]}
\renewcommand{\phi}{\varphi}
\newcommand{\eps}{\varepsilon}
\newcommand{\spacen}{\overline{S}_n}
\newcommand{\timen}{\overline{T}_n}

\begin{document}

\maketitle

\begin{abstract}
    Learned indexes leverage machine learning models to accelerate query answering in databases, showing impressive practical performance. However, theoretical understanding of these methods remains incomplete. Existing research suggests that learned indexes have superior asymptotic complexity compared to their non-learned counterparts, but these findings have been established under restrictive probabilistic assumptions. Specifically, for a sorted array with $n$ elements, it has been shown that learned indexes can find a key in $O(\log(\log n))$ expected time using at most linear space, compared with $O(\log n)$ for non-learned methods.
    
    In this work, we prove $O(1)$ expected time can be achieved with at most linear space, thereby establishing the tightest upper bound so far for the time complexity of an asymptotically optimal learned index. Notably, we use weaker probabilistic assumptions than prior \general{research}, meaning our \general{work generalizes} previous \general{results}. Furthermore, we introduce a new measure of statistical complexity for data. This metric exhibits an information-theoretical interpretation and can be estimated in practice. This characterization provides further theoretical understanding of learned indexes, by helping to explain why some datasets seem to be particularly challenging for these methods.
\end{abstract}

\section{Introduction}
\label{sec:1}
Query answering is of central importance in database systems, and developing efficient algorithms for this task is a key research problem. These algorithms often use an index structure to access records faster, at the cost of higher space usage to store the index. Arguably, the most common scenario in query answering involves point queries, where the aim is to retrieve all tuples where an attribute has an exact value \cite{gonzalez2005practical, lloyd2017near}. Also common are range queries, which ask for all tuples where an attribute is within a given interval.

Solutions based on hash tables can answer point queries efficiently \cite{cormen2022introduction}, but are ineffective in handling range queries. To generalize well to these cases, the relevant attribute is usually stored as a sorted array $A$. In this configuration, answering point and range queries essentially comes down to finding an element in a sorted array. A point query for value $q$ can be answered by finding the position of $q$ in the array if it exists. A range query matching an interval $[q, q']$ can be answered by finding the first element bigger than or equal to $q$, and then scanning the array $A$ from that position until an element greater than $q'$ is found, or the array ends.

Since the array $A$ is sorted, binary search can be used to find $q$ or the first element greater than $q$. If $A$ has $n$ elements (also called \textit{keys}), binary search takes $O(\log n)$ operations. This is asymptotically optimal in a random-access machine (RAM) model of computation, for comparison-based searching \cite{knuth1998art}. However, other index structures such as B+Trees \cite{graefe2011modern, levandoski2013bw} are preferred in practice. They have the same asymptotic complexity but make use of hardware properties to improve performance.

It was suggested in \cite{kraska2018case} that classical index structures can be understood as models that take a key as an input and predict its position in the sorted array. From this perspective, indexes such as B+Tree are potentially suboptimal because they do not take advantage of patterns in the data to improve performance. This has led to the introduction of learned indexes \cite{kraska2018case, kipf2020radixspline, ding2020alex, ferragina2020pgm, galakatos2019fiting}, which use machine learning models to predict the position of a key.

Learned indexes usually combine predictive and corrective steps, as illustrated in Figure \ref{fig:prediction-step-correction-step}. In the predictive step, a machine learning model is used to estimate the position of a key. Since this estimation may be wrong, a corrective step finds the exact position. Usually, this is done by searching around the position estimated by the model. Intuitively, the improvement in performance comes from searching over a small range, determined by the prediction error, rather than searching over the whole array. Hence, if the prediction error is small, a learned index can in principle improve upon classical methods. Most learned indexes use simple machine learning models, such as piecewise linear functions, which are fast to compute.

\begin{figure}
    \centering
    \includegraphics[width=0.5\linewidth]{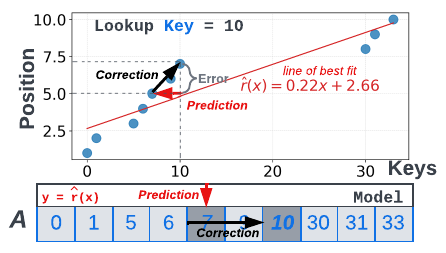}
    \caption{Predictive and corrective steps of a learned index.}
    \label{fig:prediction-step-correction-step}
    \vspace{-15pt}
\end{figure}

While learned indexes exhibit impressive experimental performance \cite{ding2020alex, galakatos2019fiting, hadian2021shift, kipf2020radixspline}, their theoretical understanding is still at an early stage. Existing research shows that specific learned indexes can achieve better asymptotic space or time complexity than traditional structures \cite{zeighami2023distribution, ferragina2020learned}, but these results rely on disparate and potentially restrictive assumptions. Furthermore, some datasets are notoriously challenging for learned indexes \cite{marcus2020benchmarking}, and no theoretical 
analysis has yet offered an adequate explanation for this.

In this work, we analyze the asymptotic expected complexity of learned indexes. Specifically, we prove that learned indexes can achieve constant expected query time with~at most linear space, constituting the best upper bound so far for the performance of an asymptotically optimal learned index. Our proof is constructive, meaning we provide a specific index that achieves this complexity.
\reviewertwo{This type of index has been considered before~\cite{zeighami2023distribution}, but we contribute tighter complexity bounds than previously available.}
Compared to prior research, our complexity bounds hold under a more general probabilistic model, and in particular, our results do not require any independence assumptions for the data generating process.

We also address the question of why certain datasets are challenging for learned indexes. It has been suggested that, in those cases, the data distribution is inherently hard to learn \cite{marcus2020benchmarking}, and several metrics have been proposed to capture this difficulty \cite{zeighami2023distribution, ferragina2020learned, wongkham2022updatable}. We introduce a new measure of data complexity $\rhohf$ with better mathematical properties. Our metric is related to Rényi entropy in information theory, helping to provide a theoretical grounding for $\rhohf$. Also, $\rhohf$ can be statistically estimated, making it potentially useful in practice to predict \textit{a priori} how well a learned index will perform. Our experiments support this idea.

\section{Preliminaries}
\label{sec:2}
\subsection{Data}
Learned indexes use statistical models. Therefore, theoretical analysis of learned indexes is done under a probabilistic model of some kind \cite{zeighami2023distribution, ferragina2020learned}. In general, keys are assumed to come from a data generating process with specified properties. For instance, in \cite{zeighami2023distribution} it is assumed that the keys are independent and identically distributed, while \cite{ferragina2020learned} uses a completely different model. We strive to generalize the probabilistic setting as much as possible. 

Formally, we model a database attribute $X$ as a stochastic process $X = \{X_i\}_{i \in \mathbb{N}}$. We assume all random variables $X_i$ have the same cumulative distribution function $F$. For any time $n\in\mathbb{N}$, we can observe the data $X_1, X_2, \ldots, X_n$ generated so far. Sorting the $\{X_i\}_{i=1}^n$ gives rise to a new set of random variables, usually called the \textit{order statistics} \cite{david2004order}, which we denote $\osx{1}, \osx{2},\ldots, \osx{n}$. We consider the sorted array $A$ at time $n$ to be formed by this last group of random variables, such that $A[i]=\osx{i}$ for all $i=1,\ldots,n$. Duplicate keys are possible in principle, and do not constitute a problem. In that case, $\osx{i+1}=\osx{i}$ for some $i$ and all analysis remains the same.

That way, in our setting a stochastic process generates the data in attribute $X$. For each $n\in\mathbb{N}$, there is a sorted array $A$ with the first $n$ keys from this process, and a learned index can be built to search this array. For now, we do not assume anything about $F$. Also, we do not assume independence of the $\{X_i\}_{i=1}^n$ variables. In Sections \ref{sec:3} and \ref{sec:4}, we analyze our results under certain probabilistic assumptions. In Section \ref{sec:5}, we compare our model and results to those of prior work.

\subsection{Learning Problem}
Assume array $A$ has $n$ keys. The learning problem can be formalized by the introduction of a $\rank$ function. For any number $q\in\mathbb{R}$, $\rank_A(q)$ is defined as the number of elements in $A$ that are smaller than or equal to $q$, that is
\begin{equation*}
    \rank_A(q)
    =
    \sum_{i=1}^n
    \indf_{\osx{i}\leq q},
\end{equation*}
where $\mathbf{1}$ is the indicator function. We omit the subscript $A$ if it is clear from context. For a point query with value $q$, $\rank(q)$ gives the position of $q$ in the array if it exists. For a range query matching an interval $[q, q']$, the relevant records are found by scanning the array $A$, starting at position $\rank(q)$ and ending when an element greater than $q'$ is found. Since $\rank$ is sufficient to answer these queries, we focus on the problem of learning the $\rank$ function under specific probabilistic assumptions.

\subsection{Model of Computation}

The model of computation we use
is the Random Access Machine (RAM) under the uniform cost criterion \cite{aho1974design}. In this model, each instruction requires one unit of time and each register (storing an arbitrary integer) uses one unit of space. This is in contrast to part of the existing literature. In \cite{zeighami2023distribution} the computation model is such that $O(\log n)$ space units are needed to store an integer $n$. On the other hand, in \cite{ferragina2020learned} an I/O model of computation is used \cite{liu2009encyclopedia}, which assumes the existence of an external memory (e.g., disk) from which data is read and written to. Under this model, time complexity consists only of the number of I/O operations.

We argue the RAM model under the uniform cost criterion is better suited for analysing learned indexes than either of these two. On the one hand, numbers are usually represented as data types with a fixed number of bits (e.g., \texttt{long int} in C++). On the other hand, most learned indexes work in main memory and avoid the use of disk storage. To make results comparable, in Section \ref{sec:5} we restate all previous complexity results using our RAM model.

Regarding space complexity, we take the standard approach of only counting the redundant space introduced by the index, and not the space necessary to store the array $A$ in the first place \cite{ferragina2020pgm}. Regarding time complexity, we do not consider the time needed to train
the model, just for finding the keys.
We focus the asymptotic analysis on \textit{expected} time, where the expectation is taken over the realizations of the stochastic process and the choice of query.


For any $n\in\mathbb{N}$, let $R_n$ be a procedure that takes the sorted array $A = [\osx{1}, \ldots, \osx{n}]$ as input, and outputs an index structure $R_n(A)$. Denote by $S(R_n(A))$ the space overhead of $R_n(A)$. Also, for any $q\in\mathbb{R}$ denote by $T(R_n(A),q)$ the query time for $q$ over array $A$ when using the index $R_n(A)$. We are interested in bounds for $S$ and $T$ that do not depend on $A$ or $q$. With that in mind, we denote
\begin{equation*}
    \spacen
    = \sup_{A = [\osx{1}, \ldots, \osx{n}]}
    S(R_n(A)),
    \quad
    \timen
    = \expec{T(R_n(A), q)},
\end{equation*}
\general{where $\mathbb{E}$ denotes expected value.}
In Section \ref{sec:3-complexity} we state asymptotic results for $\spacen$ and $\timen$.

\section{Main Results and ESPC Index}
\label{sec:3}
\subsection{Complexity Bounds}
\label{sec:3-complexity}

We mentioned above that all the $X_i$ share the same cumulative distribution function. Formally, we consider that the $X_i$ are defined in a probability space $(\Omega,\mathcal{F},\mathbb{P})$, such that $F_{X_i}(x) \coloneqq \prob{X_i\leq x}$ is the same for all $i$, and is denoted by $F$. Now, assume $F$ can be characterized by a square-integrable density function $\densityx$ and define $\rhohf$ as:
\begin{equation*}
    \rhohf
    =\norm{\densityx}_2^2
    =\int_\mathbb{R} \densityx^2(x)dx.
\end{equation*}
\reviewerone{
    Moreover, assume $f$ has bounded support $[a, b]$, that is, $f(x)=0$ for all $x\notin [a, b]$. This can be a reasonable assumption in many practical applications \cite{ma2010coding,ma2014variational} where there are known bounds to the values of data: for instance, for storing medical information such as test results, or personal financial information. Our results can be extended to the case of $f$ with unbounded support, by incurring a penalty on space usage. In Section \ref{sec:4-unbounded-support} we show how this penalty relates to the tails of the probability distribution.
}

The main complexity results are stated in Theorems \ref{thm:constant-time} and \ref{thm:loglog-time}. Both are attained by the ESPC index, which we introduce in Section \ref{sec:3-espc-index}.
\reviewerone{Both theorems as stated here assume that the query parameter follows the same probability distribution as the keys. This is done for ease of presentation.}
In Section \ref{sec:4-query-parameter} we prove that these results can be generalized to include the case when that assumption does not hold.

\begin{theorem}
    \label{thm:constant-time}
    \reviewerone{
        Suppose $f$ has support $[a, b]$ and $\rhohf < \infty$. Define  $\rho=\log\big((b-a)\rhohf\big)$. Then, there is a procedure $R_n$ for building learned indexes such that $\spacen=O(n)$ and $\timen=O(\rho)$. That is, for array $A$ with $n$ keys an index can be built with space overhead $O(n)$ and expected query time $O(\rho)$. Since $\rho$ is independent of $n$, expected time is asymptotically $O(1)$.
    }
\end{theorem}

\begin{theorem}
    \label{thm:loglog-time}
    \reviewerone{
        Suppose $f$ has support $[a, b]$ and $\rhohf < \infty$. Define  $\rho=\log\big((b-a)\rhohf\big)$. Then, there is a procedure $R_n$ for building learned indexes such that $\spacen=O\big(\frac{n}{\log n}\big)$ and $\timen=O(\rho + \log(\log n)))$. That is, for array $A$ with $n$ keys an index can be built with space overhead $O\big(\frac{n}{\log n}\big)$ and expected query time $O(\rho + \log(\log n))$. Since $\rho$ is independent of $n$, 
        expected time is asymptotically $O(\log(\log n))$.
    }
\end{theorem}

\general{We prove Theorems \ref{thm:constant-time} and \ref{thm:loglog-time} in Appendix \ref{sec:appendix-2}, as corollaries of a more general result.}


\subsection{Equal-Split Piecewise Constant Index}
\label{sec:3-espc-index}

All complexity results in Section \ref{sec:3-complexity} are achieved by the \textit{Equal-Split Piecewise Constant} (ESPC) index, which we now describe. As explained in Section \ref{sec:1}, learned indexes usually combine predictive and corrective steps. The ESPC index follows this idea, and is remarkably simple to define and implement.
\reviewertwo{It is important to mention that a similar design, the PCA Index, was analyzed in \cite{zeighami2023distribution}, but the complexity bounds we prove for this type of index are tighter and more general. Moreover, there are some design differences. First, the ESPC index uses exponential search for the corrective step, while the PCA Index uses binary search. More importantly, the PCA Index presumes knowledge of the interval $[a, b]$ where the distribution is supported, a potentially restrictive assumption that we dispense with for the ESPC index.
}

The general idea is to learn a function $\rankh$ to approximate the $\rank$ function up to some error. Then $\rank$ can be computed exactly by first evaluating $\rankh$ and then correcting the error by use of an exponential search algorithm \cite{BENTLEY197682}. If the prediction error is small, this gives a procedure for computing $\rank$ exactly in low expected time. The approximator function $\rankh$ is built as a piecewise constant function defined over equal-length subintervals. This equal-length property ensures that $\rankh$ can be evaluated in constant time.


We now describe how the approximator $\rankh$ is defined. Take an array $A$ formed by the sorted keys $\osx{1},\ldots,\osx{n}$. For a positive integer $K$, divide the range
\general{$[\osx{1},\osx{n}]$}
into $K$ subintervals of equal length, which we denote by $\{I_k\}_{k=1}^K$. Figure \ref{fig:equal-split-graph} illustrates this type of partition for a dataset with 40 keys and $K=4$. Formally, for each $k=1,\ldots,K$ the $k$-th subinterval is defined as $I_k=[t_{k-1}, t_k]$ where
\begin{equation*}
    \general{
        t_k = \osx{1} + k\delta
        \>\text{ with }\>
        \delta = \frac{\osx{n}-\osx{1}}{K}.
    }
\end{equation*}

\begin{figure}
    \centering
    \includegraphics[width=0.8\linewidth]{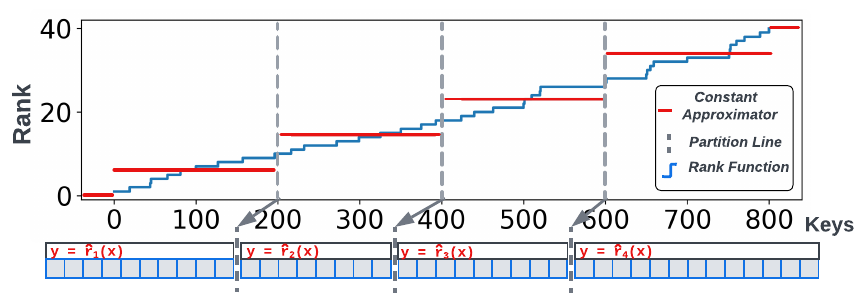}
    \vspace{-10pt}
    \caption{Partition of key range into four equal-length intervals, with approximator of $\rank$.}
    \label{fig:equal-split-graph}
    \vspace{-10pt}
\end{figure}

\general{Notice that all intervals have the same length, but not necessarily the same number of keys. Now, we want $\rankh$ to be defined as a piecewise constant function. If $q\in\mathbb{R}$ does not belong to any $I_k$, this means that either
$q<\osx{1}$ or $q>\osx{n}$, so that $\rank(q)=0$ or $\rank(q)=n$, respectively. On the other hand, for each subinterval $I_k$ the ESPC index stores a constant value $\hat{r}_k$ which approximates $\rank$ over $I_k$. Combining these ideas, we define $\rankh$ as}
\begin{equation*}
    \rankh(q) =
    \begin{cases}
        0         & \text{if } \general{q <   \osx{1}}      \\
        \hat{r}_1 & \text{if } q \in I_1    \vspace{-5pt} \\
                  & \svdots                 \\
        \hat{r}_K & \text{if } q \in I_K    \\
        n         & \text{if } \general{q >   \osx{n}}.
    \end{cases}
\end{equation*}

Figure \ref{fig:equal-split-graph} represents this type of piecewise approximator over equal-length subintervals. It remains to define the $\{\hat{r}_k\}_{k=1}^K$ values.
\general{
    Since $\rank$ is an increasing function, for all $q\in I_k$ it holds $\rank(t_{k-1}) \leq \rank(q) \leq \rank(t_{k})$. This motivates a candidate for $\hat{r}_k$ in the form of
}
\begin{equation*}
    \hat{r}_k
    = \frac{1}{2}\left(\rank(t_{k-1}) + \rank(t_{k})\right)
    = \rank(t_{k-1}) + \frac{n_k}{2},
\end{equation*}
where $n_k=\rank(t_{k}) - \rank(t_{k-1})$ is the number of keys in $I_k$. As we prove in Section \ref{sec:4-preliminary}, this means that $\hat{r}_k$ can approximate $\rank(q)$ for any $q\in I_k$ with error at most $n_k/2$, so that the quality of $\rankh$ as an approximator depends on the expected number of keys per interval.
Algorithm \ref{alg:espc-construction} details the index construction, based on our definitions for $\{I_k\}$ and $\{\hat{r}_k\}$.

\SetNlSty{text}{}{.}
\begin{algorithm}
    \SetNoFillComment
    \caption{Construction of ESPC index}
    \label{alg:espc-construction}
    \KwIn{Sorted array $A$, positive integer $K$}
    \KwOut{ESPC index for $A$ with $K$ equal-length subintervals}
    \nl $n \gets \text{length}(A)$ \\
    \nl \general{$\delta \gets (\osx{n}-\osx{1}) / K$} \\
    \nl $N \gets \text{array of length } K$ with all zeros \\
    \nl \tcp{Compute number of keys in each subinterval, store in $N$}
    \nl \For{$i=1$ \KwTo $n$}{
        \general{\nl $k \gets \ceil*{\frac{1}{\delta}(A[i] - \osx{1})}$} \\
        \nl $N[k] \gets N[k] + 1$ \\
    }
    \nl $r \gets \text{array of length } K$ \\
    \nl $r[1] \gets \frac{1}{2}N[1]$ \\
    \nl \tcp{Compute rank estimator $\hat{r}_k$ for each subinterval, store in $r$}
    \nl \lFor{$k=2$ \KwTo $K$}{
        $r[k] \gets r[k-1] + \frac{1}{2}N[k]$
    }
    \nl \KwRet{$R=(r, \delta)$}
\end{algorithm}
\vspace{-10pt}

The approximator $\rankh$ can be evaluated in constant time for any $q\in\mathbb{R}$. If
\general{$q<\osx{1}$}
then $\rankh(q)=0$, and if
\general{$q>\osx{n}$}
then $\rankh(q)=n$. If neither is true, then $q$ belongs to some interval $I_k$. In that case, $\rankh(q)=\hat{r}_k$ where $k$ is given by
\general{
    $k = \ceil*{K\frac{q-\osx{1}}{\osx{n}-\osx{1}}} = \ceil*{\frac{1}{\delta}(q-\osx{1})}.$
}
To find $\rank(q)$ exactly, the ESPC index then performs an exponential search on array $A$ around position $\ceil*{\rankh(q)}$. We prove in Section \ref{sec:4} that the expected value of the approximation error $|\rank(q)-\rankh(q)|$ is low enough to
establish the results of Section \ref{sec:3-complexity}. Algorithm \ref{alg:espc-evaluation} formalizes this evaluation process. It assumes the existence of a procedure {\normalfont \textsc{ExponentialSearch}}$(A,i,q)$ which uses exponential search around index $i$ to find the position of $q$ in array $A$. If $q$ is not in $A$, it returns the position of the greatest key lower than $q$.

\SetNlSty{text}{}{.}
\begin{algorithm}
    \SetNoFillComment
    \caption{Evaluation of ESPC index}
    \label{alg:espc-evaluation}
    \KwIn{Sorted array $A$, $q\in\mathbb{R}$, ESPC index $R=(r, \delta)$}
    \KwOut{Exact value of $\rank(q)$}
    \nl \lIf{\general{$q < \osx{1}$}}{\KwRet{$0$}}
    \nl \lElseIf{\general{$q > \osx{n}$}}{\KwRet{$n$}}
    \nl \Else{
    \nl     \general{$k \gets \ceil*{\frac{1}{\delta}(q-\osx{1})}$} \\
    \nl     $i \gets \ceil*{r[k]}$ \label{alg:espc-evaluation-step-5} \\
    \nl     \KwRet{{\normalfont \textsc{ExponentialSearch}}$(A,i,q)$} \label{alg:espc-evaluation-step-6}
    }
\end{algorithm}
\vspace{-10pt}

\section{Complexity Proofs and Analysis}
\label{sec:4}
\subsection{Preliminary Results}
\label{sec:4-preliminary}

We now turn to an analysis of the space and time complexity of the ESPC index. In the following, assume the sorted array $A$ has $n$ elements and the index is built with $K$ subintervals $\{I_k\}_{k=1}^K$. As seen from Algorithm \ref{alg:espc-construction}, an ESPC index can be represented as a tuple $R=(r, \delta)$ where $r$ is an array of length $K$ storing the approximators $\{\hat{r}_k\}_{k=1}^K$ and $\delta$ is the length of the subintervals. For $q\in\mathbb{R}$ recall that $\rankh(q)$ is the approximation for $\rank(q)$ before the exponential search (see Algorithm \ref{alg:espc-evaluation}).

In the RAM model with the uniform cost criterion, each stored value uses $1$ unit of space. Hence, the space used to store an ESPC index $R=(r, \delta)$ is $O(K)$. Additionally, as seen from Algorithm \ref{alg:espc-evaluation}, the evaluation for any $q\in\mathbb{R}$ takes constant time, plus the exponential search. An exponential search takes time $O(2\log \eps)=O(\log \eps)$ \cite{BENTLEY197682}, where $\eps$ is the difference between the starting position and the correct position. This is summarized in the following result.

\begin{proposition}
\label{prop:espc-basic-complexity}
    Let $A$ be a sorted array with $n$ entries, and $R$ be an ESPC index for $A$ with $K$ subintervals. Then, $R$ uses $O(K)$ space and given $q\in\mathbb{R}$, the ESPC index $R$ can be used to find $\rank(q)$ in $O(\log\eps)$ time, where $\eps=|\rank(q)-\rankh(q)|$.
\end{proposition}

In light of Proposition \ref{prop:espc-basic-complexity}, a way to estimate the time complexity of the ESPC index is to bound the expected value of $\eps$ as a function of $K$. Hence, our focus is on finding an expression for the approximation error $\eps=|\rank(q)-\rankh(q)|$. We start with the following result, which was mentioned in passing in Section \ref{sec:3-espc-index}.

\begin{lemma}
\label{lemma:rank-error-bound-number-keys}
    Let $k\in\{1,\ldots,K\}$ and $q\in I_k$. Then, it holds that $|\rank(q)-\rankh(q)| \leq n_k/2$ where $n_k$ is the number of keys that fall on interval $I_k$.
\end{lemma}


\general{The proof is straightforward and is deferred to Appendix \ref{sec:appendix-1}.}
\probanalysis{We now introduce the remaining necessary notation. First, notice that if $f$ has bounded support $[a, b]$, then with probability $1$ it holds that $[\osx{1}, \osx{n}]\subseteq [a,b]$. Now, for a given $K$ we can define a set of subintervals $\{J_k\}_{k=1}^K$ that partition the interval $[a, b]$ into $K$ equal-length parts, similarly to how the $\{I_k\}_{k=1}^K$ form a partition of $[\osx{1}, \osx{n}]$. See Figure \ref{fig:intervals-subintervals} for an illustration with $K=3$.
    
    From this definition, it may be the case that some $I_k$ is not entirely contained within just one $J_k$ interval, as Figure \ref{fig:intervals-subintervals} illustrates for the case of $I_2$. Notice however that if this happens, then $I_k$ must be contained in the union of two consecutive $J_k$ intervals, as exemplified by $I_2\subseteq J_1 \cup J_2$ in Figure \ref{fig:intervals-subintervals}. We formalize this observation below, which is easily verified.

    \begin{fact}
    \label{lemma:intervals-subintervals}
        Suppose $[\osx{1}, \osx{n}] \subseteq [a,b]$. Take $k,\ell\in\{1,\ldots,K\}$ such that $I_k\cap J_\ell\neq\emptyset$. Then either $I_k\subseteq J_{\ell-1}\cup J_{\ell}$ or $I_k\subseteq J_{\ell}\cup J_{\ell+1}$ holds. As a consequence, it also holds that $I_k\subseteq J_{\ell-1}\cup J_{\ell} \cup J_{\ell+1}$, where $J_0,J_{K+1}$ are defined as the empty set $\emptyset$ for completion.
    \end{fact}
}

\begin{figure}
    \begin{center}
        \begin{tikzpicture}
            \draw[line width=0.5pt](0,0)node[JB tick,label=below:$a$](1){}
                --(1.5,0)node[IB tick,label=below:$\osx{1}$](2){}
                --(3,0)node[I tick,label=below:$ $](3){}
                --(4,0)node[J tick,label=below:$ $](4){}
                --(4.5,0)node[I tick,label=below:$ $](5){}
                --(6,0)node[IB tick,label=below:$\osx{n}$](6){}
                --(8,0)node[J tick,label=below:$ $](7){}
                --(12,0)node[JB tick,label=below:$b$](8){};
            \draw[decorate, decoration={brace,raise=0.5mm}, line width=1pt, red]
                (2)--node[below,yshift=6.5mm]{$I_1$}(3);
            \draw[decorate, decoration={brace,raise=0.5mm}, line width=1pt, red]
                (3)--node[below,yshift=6.5mm]{$I_2$}(5);
            \draw[decorate, decoration={brace,raise=0.5mm}, line width=1pt, red]
                (5)--node[below,yshift=6.5mm]{$I_3$}(6);
            \draw[decorate, decoration={brace,raise=10.5mm,mirror}, line width=1pt, blue]
                (1)--node[below,yshift=-11mm]{$J_1$}(4);
            \draw[decorate, decoration={brace,raise=10.5mm,mirror}, line width=1pt, blue]
                (4)--node[below,yshift=-11mm]{$J_2$}(7);
            \draw[decorate, decoration={brace,raise=10.5mm,mirror}, line width=1pt, blue]
                (7)--node[below,yshift=-11mm]{$J_3$}(8);
        \end{tikzpicture}
    \end{center}
    \vspace*{-10pt}
    \caption{
        \probanalysis{
            Equal-length partitions for $[\osx{1}, \osx{n}]$ (given by $\{I_k\}$) and $[a, b]$ (given by $\{J_k\}$).
        }
    }
    \label{fig:intervals-subintervals}
    \vspace*{-10pt}
\end{figure}
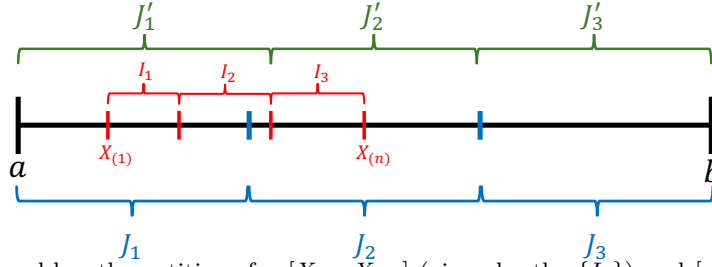


\subsection{Main Probabilistic Bound}
\label{sec:4-main-bound}

Recapping, by Proposition \ref{prop:espc-basic-complexity} the ESPC index time complexity for input $q$ depends on the approximation error $\eps=|\rank(q)-\rankh(q)|$. By Lemma \ref{lemma:rank-error-bound-number-keys}, the error $\eps$ critically depends on the number of keys in the subintervals $\{I_k\}$. We now state and prove the key results which provide bounds for the approximation error.
\reviewerone{As mentioned in Section \ref{sec:3-complexity}, we assume that the density $\densityx$ has bounded support $[a,b]$ for some $a, b\in\mathbb{R}$. That is, $\densityx(x)=0$ for any $x\not\in [a,b]$. Our results can be extended to the case of $\densityx$ with unbounded support by incurring a penalty in terms of space complexity.}
\probanalysis{Notice that the $\{J_k\}$ implicitly define a random variable, which models how the $\{X_i\}_{i=1}^n$ are distributed among the subintervals. This can be characterized by a parameter $(p_1, \ldots, p_K)$ where $p_k=\mathbb{P}(X_i\in J_k)$.}
Notice that $p_k$ is not indexed by $i$ because all $X_i$ follow the same distribution. In this context, the core of the time complexity analysis is contained in the following theorem.

\begin{theorem}
\label{thm:bound-error-K}
    \reviewerone{
        Suppose $f$ has support $[a, b]$.
    }
    Let $A$ be an array consisting of the $\{X_i\}_{i=1}^n$ sorted in ascending order. Let $\rankh$ be the approximator of the ESPC index with $K$ subintervals, as described in Section \ref{sec:3-espc-index}, and denote the approximation error for $q\in\mathbb{R}$ as $\eps(q)=|\rank(q)-\rankh(q)|$. Then, if the query parameter $q$ has the same density function $\densityx$ as the $\{X_i\}$, it holds that \probanalysis{$\expec{\eps} \leq \frac{3n}{2} \sum_{k=1}^K p_k^2$}.
\end{theorem}
\begin{proof}
    \probanalysis{
        We know that $\expec{\eps} = \expec{\expec{\eps \>|\> X_1, \ldots, X_n}}$, by the tower property for conditional expectation. The random nature of the error depends on both the choice of $q$ and the values of $\{X_i\}_{i=1}^n$. To see this, recall that the variables $X_1, \ldots, X_n$ determine the value of both $\rank$ and $\rankh$.
        The inner expected value gives
        \begin{equation*}
            \expec{\eps \>|\> X_1, \ldots, X_n}
            = 
            \int_\mathbb{R} \eps(q) \densityx(q) dq
            =
            \int_{\osx{1}}^{\osx{n}} \eps(q) \densityx(q) dq,
        \end{equation*}
        where the last equality comes from the fact that $\eps(q)=0$ if $q\notin [\osx{1}, \osx{n}]$. Replacing $\eps(q)=|\rank(q)-\rankh(q)|$ and noticing that the $\{I_k\}$ form a partition of $[\osx{1}, \osx{n}]$:
        \begin{equation*}
            \expec{\eps \>|\> X_1, \ldots, X_n}
            =
            \sum_{k=1}^K\int_{I_k} |\rank(q)-\rankh(q)| \densityx(q) dq
            \leq
            \frac{1}{2}\sum_{k=1}^K n_k \int_{I_k} \densityx(q) dq,
        \end{equation*}
        where $n_k$ is the number of keys in subinterval $I_k$ and we have used Lemma \ref{lemma:rank-error-bound-number-keys} to bound the error within each $I_k$. Now, denote by $n_{k,\ell}$ the number of keys in $I_k \cap J_\ell$. Then, we can write
        \begin{equation}
        \label{eq:conditional-error-in-terms-of-I}
            \expec{\eps \>|\> X_1, \ldots, X_n}
            \leq
            \frac{1}{2}\sum_{k=1}^K \left(\sum_{\ell = 1}^K n_{k, \ell}\right) \int_{I_k} \densityx(q) dq
            =
            \frac{1}{2}\sum_{\ell=1}^K \sum_{k = 1}^K n_{k, \ell} \int_{I_k} \densityx(q) dq,
        \end{equation}
        where first we have written $n_k$ as $\sum_{\ell = 1}^K n_{k, \ell}$ and then we have exchanged the order of summation. Now, denote by $L(k)$ the set of $\ell$ indices such that $I_k\cap J_\ell \neq \emptyset$. Notice that $n_{k, \ell}=0$ if $\ell\notin L(k)$. Hence, inequality (\ref{eq:conditional-error-in-terms-of-I}) becomes:
        \begin{equation}
        \label{eq:conditional-error-in-terms-of-I-and-J}
            \expec{\eps \>|\> X_1, \ldots, X_n}
            \leq
            \frac{1}{2}\sum_{\ell=1}^K \sum_{k: \ell\in L(k)} n_{k, \ell} \int_{I_k} \densityx(q) dq.
        \end{equation}
        By Fact \ref{lemma:intervals-subintervals} we know that $I_k\subseteq J_{\ell-1} \cup J_{\ell} \cup J_{\ell+1}$ for all $\ell$ and $k$ such that $\ell\in L(k)$. Using the notation $p_\ell = \prob{X_i\in J_\ell}=\int_{J_\ell} \densityx(q) dq$ introduced above, from inequality (\ref{eq:conditional-error-in-terms-of-I-and-J}) we get
        \begin{equation}
        \label{eq:conditional-error-in-terms-of-J}
            \expec{\eps \>|\> X_1, \ldots, X_n}
            \leq
            \frac{1}{2}
            \sum_{\ell=1}^K (p_{\ell-1}+p_{\ell}+p_{\ell+1})
            \sum_{k: \ell\in L(k)} n_{k, \ell}
            =
            \frac{1}{2} \sum_{\ell=1}^K (p_{\ell-1}+p_{\ell}+p_{\ell+1}) m_\ell,
        \end{equation}
        where $m_\ell$ denotes the number of keys in $J_\ell$. Now, applying expected value on both sides of inequality (\ref{eq:conditional-error-in-terms-of-J}), we get that $\expec{\eps}=\expec{\expec{\eps \>|\> X_1, \ldots, X_n}} \leq \frac{1}{2}\sum_{\ell=1}^K (p_{\ell-1}+p_{\ell}+p_{\ell+1}) \expec{m_\ell}$. Furthermore, we have
        \begin{equation*}
            \expec{m_\ell}
            = \expec{\sum_{i=1}^n \mathbf{1}_{X_i\in J_\ell}}
            = \sum_{i=1}^n \expec{\mathbf{1}_{X_i\in J_\ell}}
            = \sum_{i=1}^n \prob{X_i\in J_\ell}
            = n p_\ell.
        \end{equation*}
        Replacing in the expression for $\expec{\eps}$, we get the following inequality:
        \begin{equation}
        \label{eq:expected-error-in-terms-of-p}
            \expec{\eps}
            \leq
            \frac{1}{2}\sum_{\ell=1}^K (p_{\ell-1}+p_{\ell}+p_{\ell+1}) \expec{m_\ell}
            =
            \frac{n}{2}\sum_{\ell=1}^K (p_{\ell-1}p_{\ell}+p_{\ell}^2+p_{\ell}p_{\ell+1}).
        \end{equation}
        Now, using Cauchy–Schwarz inequality and the fact that $p_0=\prob{X_i\in J_0}=0$ we get:
        \begin{equation*}
            \sum_{\ell=1}^K p_{\ell-1}p_{\ell}
            \leq
            \left(\sum_{\ell=1}^K p_{\ell-1}^2 \sum_{\ell=1}^K p_{\ell}^2\right)^{1/2}
            =
            \left(\sum_{\ell=1}^{K-1} p_{\ell}^2 \sum_{\ell=1}^K p_{\ell}^2\right)^{1/2}
            \leq
            \left(\sum_{\ell=1}^K p_{\ell}^2 \sum_{\ell=1}^K p_{\ell}^2\right)^{1/2}
            =
            \sum_{\ell=1}^K p_{\ell}^2.
        \end{equation*}
        Similarly, using Cauchy–Schwarz inequality and the fact that $p_{K+1}^2=0$ we can show that $\sum_{\ell=1}^K p_{\ell}p_{\ell+1}\leq \sum_{\ell=1}^K p_{\ell}^2$. Finally, applying these inequalities in (\ref{eq:expected-error-in-terms-of-p}), we get the desired result:
        \begin{equation*}
            \expec{\eps}
            \leq
            \frac{n}{2}\sum_{\ell=1}^K (p_{\ell-1}p_{\ell}+p_{\ell}^2+p_{\ell}p_{\ell+1})
            \leq
            \frac{n}{2}\left(\sum_{\ell=1}^K p_{\ell}^2+\sum_{\ell=1}^K p_{\ell}^2+\sum_{\ell=1}^K p_{\ell}^2\right)
            \leq
            \frac{3n}{2}\sum_{\ell=1}^K p_{\ell}^2.
        \end{equation*}
    }
\end{proof}

Theorem \ref{thm:bound-error-K} is the key tool for proving our results, as we show in Section \ref{sec:4-proof-complexity}.

\subsection{Proof of Asymptotic Complexity}
\label{sec:4-proof-complexity}

We now use Theorem \ref{thm:bound-error-K} to prove the results stated in Section \ref{sec:3-complexity}. For this, we need to relate the error bound in Theorem \ref{thm:bound-error-K} with $\rhohf=\norm{\densityx}_2^2$. For this, notice that
\general{
    \begin{equation*}
        p_k^2
        =
        \left(\int_{J_k} \densityx(x) dx\right)^2
        \leq
        \left(\int_{J_k} 1 dx\right)
        \left(\int_{J_k} \densityx^2(x) dx\right)
        =
        |J_k| \int_{J_k} \densityx^2(x) dx,
    \end{equation*}
    where $|J_k|$ denotes the length of interval $J_k$ and we have used Cauchy-Schwarz inequality to bound $\big(\int_{J_k} \densityx(x) dx\big)^2$. Substituting this in the error bound from Theorem \ref{thm:bound-error-K} and using the fact that $|J_k|=(b-a)/K$ for all $k$, we obtain
    \begin{align*}
        \expec{\eps}
        &\leq
        \frac{3n}{2}\sum_{k=1}^K p_k^2
        \leq
        \frac{3n}{2}\sum_{k=1}^K
        \frac{(b-a)}{K} \int_{J_k} \densityx^2(x) dx
        =
        \frac{3(b-a)}{2}\frac{n}{K}\sum_{k=1}^K \int_{J_k} \densityx^2(x) dx
    \intertext{Since $\{J_k\}$ is a partition of $[a, b]$ and $\densityx$ is square-integrable:}
        \expec{\eps}
        &\leq
        \frac{3(b-a)}{2}\frac{n}{K}\int_a^b \densityx^2(x) dx
        =
        \frac{3(b-a)}{2}\frac{n}{K} \int_\mathbb{R} \densityx^2(x) dx
        =
        \frac{3(b-a)}{2}\rhohf\frac{n}{K}.
    \end{align*}
    
    We can summarize the main error bound for the ESPC index in the following proposition.}

\reviewerone{\begin{proposition}
    \label{prop:bound-error-rho}
        Under the conditions of Theorem \ref{thm:bound-error-K}, and assuming $\rhohf < \infty$, it holds
        \begin{equation*}
            \expec{\eps}
            \leq
            \frac{3(b-a)}{2}\rhohf\frac{n}{K}.
        \end{equation*}
    \end{proposition}}
   
\general{The complexity results from Section \ref{sec:3-complexity} are proved in Appendix \ref{sec:appendix-2} as corollaries of Proposition \ref{prop:bound-error-rho}. In Sections \ref{sec:4-query-parameter} and \ref{sec:4-unbounded-support} we discuss two important generalizations of our results. In Section \ref{sec:5}, we present further analysis and a comparison to previous results.}

\subsubsection{Distribution of query parameter}
\label{sec:4-query-parameter}

As can be seen from the statement of Theorem \ref{thm:bound-error-K}, we have assumed that $q$ has the same probability density function $\densityx$ as the $\{X_i\}$. We believe this to be a reasonable approximation in many contexts, but our analysis can be extended to $q$ otherwise distributed, which constitutes an important generalization with respect to previous work \cite{zeighami2023distribution}.

\reviewerone{
    \begin{proposition}
    \label{prop:bound-error-rho-query}
        Assume the conditions of Theorem \ref{thm:bound-error-K}. Additionally, suppose $\rhohf < \infty$ and that the search parameter $q$ distributes according to a density function $g$. Further suppose that $g$ is square-integrable and denote $\rhohg = \norm{g}_2^2$. Then, it holds that
        \begin{equation*}
            \expec{\eps}
            \leq
            \frac{3(b-a)}{2}\sqrt{\rhohf\rhohg}\frac{n}{K}.
        \end{equation*}
    \end{proposition}
}

See Appendix \ref{sec:appendix-3} for the proof. As can be seen, for the general case where $q$ does not distribute according to $f$, our results have essentially the same form. For ease of presentation, all our analyses from this point on assume the case where $g=f$. By Proposition \ref{prop:bound-error-rho-query} we know this does not result in any loss of 
generality. For the case of $q$ otherwise distributed, we substitute $\sqrt{\rhohf\rhohg}$ in place of $\rhohf$.

\subsubsection{Unbounded support}
\label{sec:4-unbounded-support}

\reviewerone{
    Proposition \ref{prop:bound-error-rho} can be extended to $\densityx$ with unbounded support. For $a, b\in\mathbb{R}$ denote by $A_{a, b}$ the event that $[\osx{1}, \osx{n}]\not\subseteq [a,b]$. Take values of $a,b$ such that $\prob{A_{a, b}}\leq \frac{1}{\log n}$. Then, the expected query time is $O(\expec{\log \eps})$ where the expected value can be written as
    \begin{equation*}
        \expec{\log \eps}
        =
        E_1 + E_2
        \>\text{ where }\>
        E_1 = \expec{\log \eps \>\big|\> A_{a, b}^C}\prob{A_{a, b}^C}
        \>\text{ and }\>
        E_2 = \expec{\log \eps \>\big|\> A_{a, b}}\prob{A_{a, b}}.
    \end{equation*}
    Since the prediction error $\eps$ can never exceed $n$, the $E_2$ term can be bounded as
    \begin{equation*}
        E_2
        = \expec{\log \eps \>\big|\> A_{a, b}}\prob{A_{a, b}}
        \leq \log n \cdot \prob{A_{a, b}}
        \leq \log n \cdot \frac{1}{\log n} = 1.
    \end{equation*}
    On the other hand, the event $A_{a, b}^C$ means that $[\osx{1}, \osx{n}]\subseteq [a, b]$. This allows us to apply Jensen's inequality along with Proposition \ref{prop:bound-error-rho} to get
    \begin{equation*}
        E_1
        =
        \expec{\log \eps \>\big|\> A_{a, b}^C}\prob{A_{a, b}^C}
        \leq
        \expec{\log \eps \>\big|\> A_{a, b}^C}
        \leq
        \log \expec{\eps \>\big|\> A_{a, b}^C}
        \leq
        \log \left(\frac{3(b-a)}{2}\rhohf\frac{n}{K}\right).
    \end{equation*}
    Putting everything together, this means that expected query time is
    \begin{equation}
    \label{eq:expected-log-error-rho-n-K}
        O(\expec{\log \eps})
        =
        O\left(\log \left(\frac{3(b-a)}{2}\rhohf\frac{n}{K}\right) + 1\right)
        =
        O\left(\log \left((b-a)\rhohf\frac{n}{K}\right)\right).
    \end{equation}
    Hence, time complexity depends on the values of $a, b$ that guarantee $\prob{A_{a, b}}\leq \frac{1}{\log n}$. These values depend on how heavy-tailed the distribution is. We propose the following two variations of Theorem \ref{thm:constant-time}, with other results being available based on different assumptions.
}

\reviewerone{
    \begin{theorem}
        \label{thm:constant-time-unbounded-1}
            Suppose the $\{X_i\}$ have finite mean $\mu$ and variance $\sigma^2$, and that $\rhohf < \infty$. Then, there is a procedure $R_n$ for building learned indexes with $\spacen=O\big(n\sqrt{n\log n}\big)$ and $\timen=O(\rho)$, where $\rho=\log\big(2\sigma\rhohf\big)$. Since $\rho$ is independent of $n$, expected time is asymptotically $O(1)$.
    \end{theorem}
}

\reviewerone{
    \begin{theorem}
        \label{thm:constant-time-unbounded-2}
                Suppose $\rhohf < \infty$ and that the distribution of the $\{X_i\}$ is subexponential~\cite{Vershynin_2018}, meaning there is a constant $C>0$ such that $\prob{|X_i| \geq x} \leq 2e^{-Cx}$ for all $x\geq 0$. Then, there is a procedure $R_n$ for building learned indexes with $\spacen=O(n \log n)$ and $\timen=O(\rho)$, where $\rho=\log\big(4\rhohf/C\big)$. Since $\rho$ is independent of $n$, expected time is asymptotically $O(1)$.
    \end{theorem}

Both Theorems are proved in Appendix \ref{sec:appendix-4}. In particular, Theorem \ref{thm:constant-time-unbounded-2} can also be viewed as a corollary of Lemma 3.4 in \cite{zeighami2023distribution}, which provides an alternative proof for this result.}




\section{Analysis and Benchmarking}
\label{sec:5}
We now take a closer look at the expected time complexity. Consider an ESPC index with $K$ subintervals for a sorted array $A$ with $n$ keys. Then, by Theorem \ref{thm:bound-error-K} and Jensen's inequality:
\begin{equation}
\label{eq:time-bound-discrete-probability}
    \expec{\log \eps}
    \leq
    \log \expec{\eps}
    \leq
    \log \left( \frac{3n}{2}\sum_{k=1}^K p_k^2 \right)
    =
    \log \left(\frac{3n}{2}\right) + \log \left( \sum_{k=1}^K p_k^2 \right)
\end{equation}
\general{
    where $p_k = \mathbb{P}\big(X_i\in J_k\big)$. Recall that $(p_1, \ldots, p_K)$ constitutes an induced discrete probability distribution, describing the probabilities of the $\{X_i\}$ falling on the different $\{J_k\}$ intervals.
}

\subsection{Analysis of Alternative Index Design}
\label{sec:5-alternative-design}

\general{
    The term $\log \big( \sum_{k=1}^K p_k^2 \big)$ is related to a notion of entropy and is minimized when $(p_1, \ldots, p_K)$ is a uniform distribution (see Section \ref{sec:5-renyi}). Hence, we could conceivably get better performance by trying to induce $p_k=1/K$ for all $k$. Consider then an alternative design where the $\{I_k\}$ subintervals are defined in such a way that they cover the $[a, b]$ range (instead of the $[\osx{1},\osx{n}]$ range) and we aim to enforce $p_k=1/K$ for all $k$, rather than an equal-length condition. This requires knowledge of $a$, $b$, and $F$, which we usually lack, but several strategies can be explored to estimate them from the data.
}

This new index design would refine the intervals where we expect data to be more concentrated, which seems to be a reasonable strategy. However, there is a trade-off. Choosing the subintervals so that $p_k=1/K$ means that they will not be equal-length, unless the $\{X_i\}$ distribute uniformly. Consequently, for a given query parameter $q$, we can no longer find its subinterval in constant time, as we did before with $k \gets \ceil*{\frac{1}{\delta}\big(q-\osx{1}\big)}$. This version of the index would need a new strategy to find the relevant subinterval $I_k$ such that $q\in I_k$.

\general{
    A standard strategy would be to create a new array $A'$ formed by the left endpoints of the $\{I_k\}$, that is, $A'=[t_0, \ldots, t_{k-1}]$. Now, for a given search parameter $q$ we can find the relevant $I_k$ with $k=\rank_{A'}(q)$, and then proceed with steps (\ref{alg:espc-evaluation-step-5}) and (\ref{alg:espc-evaluation-step-6}) from Algorithm \ref{alg:espc-evaluation}. Notice that the index is now hierarchical, with $A'$ at the top and $A$ (segmented into the $\{I_k\}$) at the bottom. From equation (\ref{eq:time-bound-discrete-probability}), expected time at the bottom layer would scale as
}
\begin{equation*}
    \expec{\log \eps_\text{bot}}
    \leq
    \log \left(\frac{3n}{2}\sum_{k=1}^K p_k^2 \right)
    =
    \log \left(\frac{3n}{2}\sum_{k=1}^K \frac{1}{K^2} \right)
    =
    \log \left(\frac{3n}{2K}\right).
\end{equation*}

On the other hand, the top layer has $K$ keys $\{t_0,\ldots,t_{k-1}\}$. Since the index is designed so that $p_k=1/K$, this means $F(t_i)= i/k$. As a consequence, the keys in $A'$ approximately follow distribution $F$, with a better approximation as $K\to\infty$. Suppose we use an ESPC index with $K'$ subintervals to speed up computation of $\rank_{A'}$. By Proposition \ref{prop:bound-error-rho} the expected time at the top layer will be dominated by $\expec{\log \eps_\text{top}} \leq \log \left(3(b-a)\rhohf K/(2K') \right)$ and the total expected time can be bounded as
\general{
    \begin{equation*}
        \expec{\log \eps_\text{bot}}
        + \expec{\log \eps_\text{top}}
        \leq
        \log \left(\frac{3n}{2K}\right)
        + \log \left(\frac{3(b-a)}{2}\rhohf \frac{K}{K'}\right)
        =
        \log \left(\frac{9(b-a)}{4}\rhohf \frac{n}{K'}\right).
    \end{equation*}
}

\general{
    Notice that $(b-a)\rhohf$ is still present in this bound. Now, for instance, with $K=K'=n$ we get twice the space usage compared to Theorem \ref{thm:constant-time} and the same asymptotic expected query time. This indicates that asymptotic complexity remains the same for a hierarchical or tree-like version of the index, so the basic ESPC index gives a simpler proof for our results.
}

\subsection{Bounds for Expected Query Time}

\subparagraph{\general{Complexity due to number of keys.}}

Equation (\ref{eq:time-bound-discrete-probability}) provides an interesting insight. We see that the expected log-error (and hence, the expected query time) can be decomposed into two separate sources of complexity. The first is due exclusively to the number of keys $n$, which suggests that the query problem gets increasingly difficult when new keys are inserted, even if the underlying distribution of the data remains the same. This helps to explain why learned indexes evidence a need for updating \cite{yang2023flirt, shahvarani2020parallel} in the face of new keys arriving, even when the same dataset is used to simulate the initial keys and the subsequent arrivals.

\subparagraph{\general{Information-theoretical considerations.}}
\label{sec:5-renyi}

The sum $\sum_{k=1}^K p_k^2$ is related to the concept of Rényi entropy \cite{renyi1961measures}. For a discrete random variable $P$ with possible outcomes ${1,\ldots,K}$ with correspondent probabilities $\{p_k\}_{k=1}^K$, the Rényi entropy of order $\alpha$ is defined as
\begin{equation*}
    H_\alpha(P)
    = \frac{1}{1-\alpha}
    \log \textstyle{\left(\sum_{k=1}^K p_k^\alpha \right)}.
\end{equation*}
for $\alpha > 0, \alpha\neq 1$. The cases $\alpha=0,1,\infty$ are defined as limits, with $H_1(P)$ giving the Shannon entropy \cite{renyi1961measures}. The Rényi entropy can be understood as a generalization of the Shannon entropy which preserves most of its properties.
Consequently, equation (\ref{eq:time-bound-discrete-probability}) can be written as:
\begin{equation}
    \label{eq:time-bound-entropy-discrete}
        \general{
            \expec{\log \eps}
            \leq \log\left(\frac{3n}{2}\right) - H_2(P).
        }
    \end{equation}
\general{$H_2$ appears in cryptography \cite{cachin1997entropy, skorski2015shannon} due to its relation with the collision probability $\sum_{k=1}^K p_k^2$, and is maximized by the uniform distribution \cite{skorski2015shannon}, that is, when $p_k=1/K$ for all $k$.}
\general{This suggests that negative entropy might be a good measure of statistical complexity for learned indexes, beyond the ESPC index. It captures a notion of distance to the uniform distribution $U$, which is easy to learn for the type of simple models commonly used in learned indexes. This notion of distance can be formalized by writing $-H_2(P)=D_2(P||U)-\log K$, where $D_2(P||U)$ is the Rényi divergence of order 2, defined analogously to the KL divergence \cite{van2014renyi}.}

\subparagraph{\general{Characterization via the stochastic process.}}

A similar insight can be reached by reference to the density function $\densityx$. By Proposition \ref{prop:bound-error-rho} and Jensen's inequality, we know
\reviewerone{
    \begin{equation*}
        \expec{\log \eps}
        \leq
        \log \left( (b-a)\rhohf\frac{n}{K} \right)
        =
        \log n + \log (b-a) - \log (K) - h_2(X_i),
    \end{equation*}
}
where we use $h_\alpha(X)=\frac{1}{1-\alpha}\log (\norm{\densityx}_\alpha^\alpha)$ to denote the Rényi entropy for continuous random variables \cite{vinga2004renyi}, similar to how differential entropy is defined in analogy to the Shannon entropy.
\general{This last expression provides a useful bound for the expected log-error, by separating the effect of the number $n$ of keys, the number $K$ of subintervals (which represent space usage), the support $b-a$, and an inherent characteristic $h_2(X_i)$ of the stochastic process.}



\subsection{Comparison with Existing Methods}
\label{sec:5-comparison}


In this section, we compare our probabilistic model and complexity bounds to previous work. All results are stated in terms of our RAM model. Table \ref{tab:complexity-benchmark} summarizes the main points.

\begin{table}
    \centering
    \begin{tabular}{|l|l|l|l|l|} 
         \hline
         \textbf{Method}
            & \textbf{Space}
            & \textbf{Expected time}
            & \makecell[ll]{
                \textbf{Parameters / } \\
                \textbf{Notation}
            }
            & \textbf{Probabilistic Model} \\
         \hline
         B+Tree
            & $O\left(\frac{n}{b}\right)$
            & $O(\log n)$
            & \makecell[ll]{
                Branching \\ factor $b$
            }
            & None \\
         \hline
         PGM \cite{ferragina2020pgm}
            & $O\left(\frac{n}{\eps^2}\right)$
            & $O(\log n)$
            & \makecell[ll]{
                Maximum \\ error $\eps$
            }
            & \makecell[ll]{
                Gaps $\{X_{i+1}-X_i\}$ i.i.d. \\
                with finite mean $\mu$ and \\
                variance $\sigma^2$, $\eps\gg \sigma/\mu$
             } \\
         \hline
         PCA \cite{zeighami2023distribution}
             & $O\left(\reviewerone{\rho}^{1+\delta} n^{1+\frac{\delta}{2}}\right)$
             & $O\left(\log\frac{1}{\delta}\right)$
             & \makecell[ll]{
                $\delta > 0$, \\
                \reviewerone{$\rho = (b-a)\rho_1$}
            }
             & \makecell[ll]{
                $\{X_i\}$ i.i.d. with
                density $\densityx$, \\
                $\densityx(x) \leq \rho_1 < \infty$ for all $x$
             } \\
         \hline
         RDA \cite{zeighami2023distribution}
            & $O\left(\general{\rho} n\right)$
            & $O(\log(\log n))$
            & \general{$\rho = \frac{\rho_1}{\rho_2}$}
            & \makecell[ll]{
                Same as PCA Index, plus \\
                $0 < \rho_2 \leq \densityx(x)$ for all $x$
             } \\
         \hline
         ESPC
            & \makecell[ll]{
                $O(n)$ \\ [1ex]
                $O\big( \frac{n}{\log n}\big)$
            }
            & \makecell[ll]{
                $O(\reviewerone{\rho})$ \\ [1ex]
                $O(\reviewerone{\rho} + \log(\log n))$
            }
            & \makecell[ll]{
                $K$ intervals, \\
                \reviewerone{$\rho = \log\big((b-a)\rhohf\big)$}
            }
            & \makecell[ll]{
                Variables $\{X_i\}$ have same \\
                density $\densityx$, $\rhohf = \norm{\densityx}_2^2 < \infty$
            } \\
         \hline
    \end{tabular}
    \vspace{2pt}
    \caption{Space and expected time complexity for learned indexes and B+Tree. The probabilistic assumptions are not necessary for the index to work, but they are necessary to guarantee the stated performance. \reviewerone{For the last three rows, it is assumed that $f$ has bounded support $[a, b]$.}}
    \label{tab:complexity-benchmark}
    \vspace{-23pt}
\end{table}

\subparagraph{\general{Non-learned methods.}}

There are many classical algorithms for searching a sorted array. They tend to be comparison-based and have $O(\log n)$ expected time, unless strong assumptions are made (e.g., interpolation search uses $O(\log(\log n))$ time on average if the keys are uniformly distributed \cite{knuth1998art}). We use B+Tree as a representative. A B+Tree has space overhead $O(n/b)$ where $b$ is the branching factor. Since $b$ does not depend on $n$, asymptotically this is just $O(n)$. The number of comparisons required to find a key is $O(\log n)$. This represents both worst-case and average-case complexity, with no probabilistic assumptions.

\subparagraph{\general{PGM-index.}}

In the absence of probabilistic assumptions, space overhead of the PGM-index is $O(n/\eps)$ and expected time is $O(\log n)$. Here, $\eps$ is a hyperparameter of the method and is independent of $n$, so this constitutes the same asymptotic complexity as B+Trees. Under certain assumptions, space overhead can be strengthened to $O(n/\eps^2)$ \cite{ferragina2020learned}, representing a constant factor improvement (i.e., same asymptotic space complexity).

In contrast to \cite{zeighami2023distribution} and our work, the assumptions in \cite{ferragina2020learned} concern the \textit{gaps} between keys. Specifically, the $\{X_{i+1}-X_i\}$ are assumed to be independent and identically distributed (i.i.d.), with finite mean $\mu$ and variance $\sigma^2$. It is also assumed that $\eps \gg \sigma/\mu$.  This setting is not directly comparable to ours, in the sense that neither is a particular instantiation of the other. However, the assumptions described above are very strong in our view and seem hard to justify for most practical applications. Also, the $\eps \gg \sigma/\mu$ condition constrains the choice of this key hyperparameter.

\subparagraph{\general{PCA Index and RDA Index.}}

    In \cite{zeighami2023distribution} several indexes are presented, the most relevant for us being the PCA and RDA indexes. \reviewertwo{The PCA Index is similar to the ESPC index, with some differences as explained in Section \ref{sec:3-espc-index}, including the use of binary search instead of exponential search.
}The RDA Index is similar to the Recursive Model Index (RMI) from \cite{kraska2018case}. All results below are stated in terms of the RAM model under the uniform cost criterion.

\reviewerone{
    In \cite{zeighami2023distribution}, the same assumption regarding bounded support $[a, b]$ for $f$ is needed, so we consider that case in this section. Extensions to unbounded support incur a penalty in terms of space complexity, dependent on the tails of the distribution. We first describe the PCA Index. For any $\delta>0$, there is a PCA Index with space overhead $O((b-a)^{1+\delta}\rho_1^{1+\delta}n^{1+\delta/2})$ and expected time complexity $O(\log\frac{1}{\delta})$.}
Here, $\rho_1$ is an upper bound for $\densityx$. Since it is independent of $n$, space overhead is asymptotically $O(n^{1+\delta/2})$. A direct comparison to the ESPC index is possible through Theorem \ref{thm:constant-time}, which provides a strictly stronger result, with expected constant time and truly linear space.

On the other hand, the RDA Index has space overhead $O({\scriptstyle\frac{\rho_1}{\rho_2}}n)$, where $\rho_2$ is a lower bound for $\densityx$ and $\rho_1$ is an upper bound (as for the case of the PCA Index). Asymptotically in $n$, this is $O(n)$. Expected query time in this case is $O(\log(\log n))$. A direct comparison to the ESPC index is possible through Theorem \ref{thm:loglog-time}, which provides a strictly stronger result, with $O(\log(\log n))$ expected time and sublinear space overhead.

Even as we improve the complexity bounds, we also generalize the domain of application for our result. As part of the probabilistic model for the PCA and RDA indexes, it is assumed that the $\{X_i\}$ are i.i.d. with density $\densityx$ \cite{zeighami2023distribution}. PCA index assumes $\densityx$ is upper bounded, that is, there exists some fixed $\rho_1<\infty$ such that $\densityx(x)\leq\rho_1$ for all $x$. Notice that this condition implies that $\densityx$ is square-integrable, so that our assumptions are strictly weaker:
\begin{equation*}
    \rhohf
    =
    \int_\mathbb{R} \densityx^2(x) dx
    \leq \rho_1 \int_\mathbb{R} \densityx(x) dx
    = \rho_1
    < \infty.
\end{equation*}

\general{The RDA index further assumes that $\densityx$ is bounded away from $0$, that is, there exists some $\rho_2>0$ such that $\densityx(x) \geq \rho_2$.}
\reviewertwo{Moreover, as described in \cite{zeighami2023distribution}, access to the values of $a$, $b$ and $\rho_1$ is required to build the PCA Index, while the RDA Index requires knowledge of $\rho_1$ and $\rho_2$. This seems like a strong assumption in practice.}
\reviewerone{Finally, the analysis in \cite{zeighami2023distribution} holds less generally than ours because it requires the $\{X_i\}$ to be independent, and it (implicitly) depends on the query parameter $q$ having the same distribution as the $\{X_i\}$.
}

\section{Experimental Findings}
\label{sec:6}
We implement the ESPC index and evaluate the results in light of the theoretical analysis.

\subsection{Data}

In our experiments, we use four datasets from the \textit{Searching on Sorted Data} (SOSD) benchmark \cite{kipf2019sosd,marcus2020benchmarking}, which has become standard in testing learned indexes \cite{sun2023learned}. Each dataset consists of a sorted array of keys, which can be used to build an index and simulate queries. Two of the datasets we use (\textit{usparse}, \textit{normal}) are synthetically generated. The other two (\textit{amzn}, \textit{osm}) come from real-world data. In particular, \textit{osm}
is known to be challenging for learned indexes \cite{marcus2020benchmarking}.

\subsection{Experimental Design}

Our complexity bounds (Theorems \ref{thm:constant-time} and \ref{thm:loglog-time}) are corollaries of Propositions \ref{prop:espc-basic-complexity} and \ref{prop:bound-error-rho}, which are the key results to validate empirically. We do this in the following way. Each dataset consists of a sorted array of $20\times 10^7$ keys. We use a subsample of $n=10^7$ keys, so that experiments are less computationally expensive. For a fixed subsample, we take different values of $K$ and for each:

\begin{enumerate}
    \item We build the ESPC index with $K$ subintervals. Space overhead is estimated by reporting the amount of memory used by the index.
    From Proposition \ref{prop:espc-basic-complexity}, this should be $O(K)$. In terms of experiments, we expect a linear relationship between memory overhead and $K$.
    \item We sample $Q$ keys from $A$. We compute $\rank$ for these keys using the index and measure the prediction error. We then average these values to estimate the expected error.
    \general{From Proposition \ref{prop:bound-error-rho}, this should be smaller than $\frac{3(b-a)}{2}\rhohf\frac{n}{K}$. In practice, we use an estimate $\rhohat$ of $\rhohf$ (see Section \ref{sec:6-rho-analysis}) and all data is re-scaled to the $[0, 1]$ interval, so that we consider $b-a=1$. For experiments, we plot the average prediction error alongside $(3\rhohat n)/(2K)$, as functions of $K$. We expect the average error curve to be below $(3\rhohat n)/(2K)$.}
\end{enumerate}

We use $Q=30 \times 10^6$. For $K$ we use values of $10^3, 5\times 10^3, 10^4, 5\times 10^4, 10^5$ and $2\times 10^5$.

\subsection{Results}
\label{sec:6-results}

In terms of storage, our experiments show a perfect linear relationship between memory overhead and $K$, where $\texttt{memory overhead} = 32K \texttt{ bytes}$. This result is in accordance to the $O(K)$ estimate in Proposition 4.1. The agreement between experiments and theory is to be expected here, because the space overhead of the ESPC index does not depend on probabilistic factors.

On the other hand, Figure \ref{fig:error-experiments-theory} shows how the prediction error changes with the number of subintervals $K$. We plot the average experimental error along with the theoretical bound for the expected error. In accordance with Proposition \ref{prop:bound-error-rho}, this theoretical bound is given by \general{$(3n \rhohat)/(2K)$}. As expected, for all datasets this expression serves as an upper bound for the average experimental error, and both curves have a similar shape. The results show that Proposition \ref{prop:bound-error-rho} has good predictive power, even for real-world data.

\begin{figure}
    \centering
    \includegraphics[width=0.65\linewidth]{
        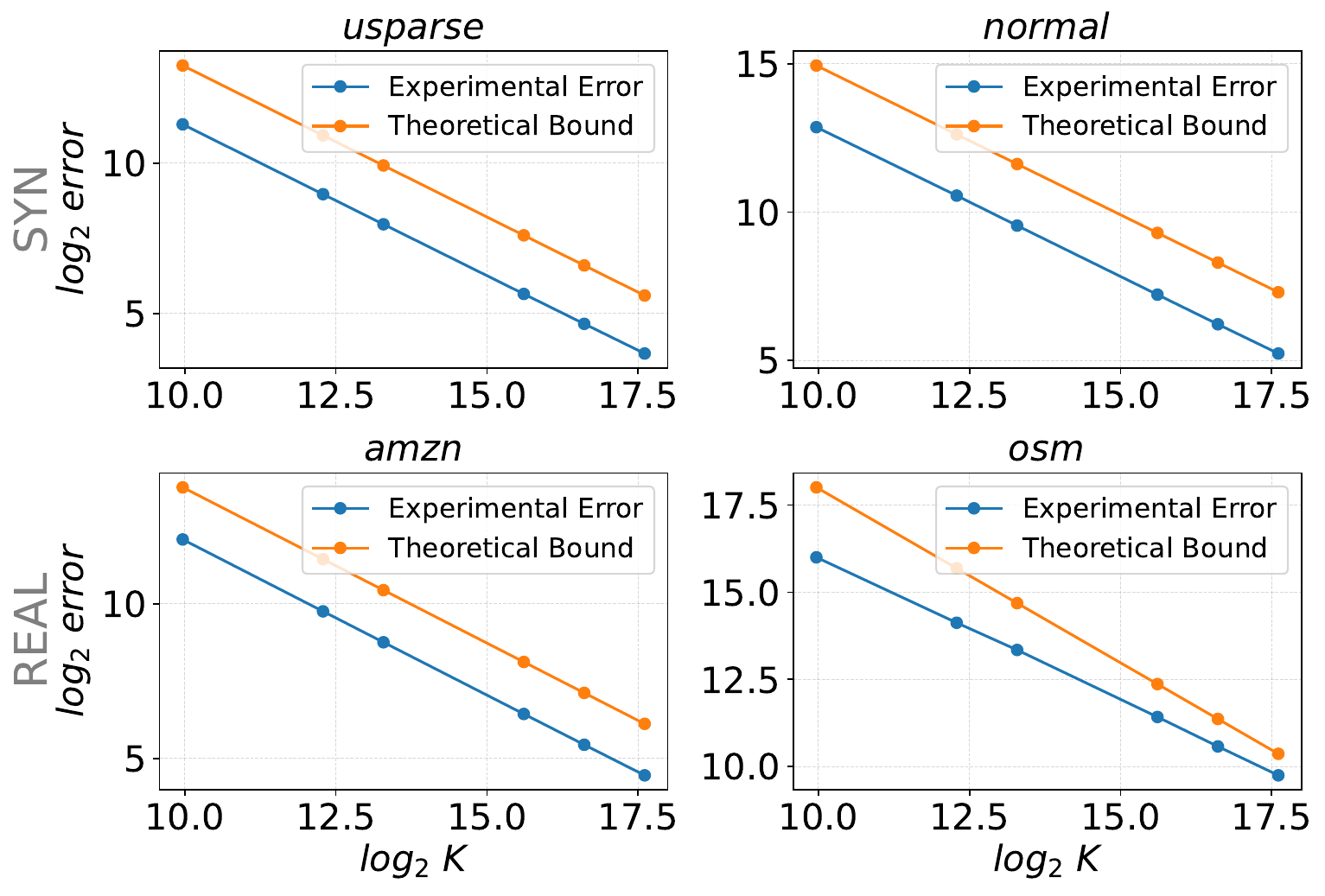
    }
    \caption{Average experimental error and theoretical bound for expected error, as functions of the number $K$ of subintervals. Top row shows synthetic datasets, bottom row shows real-world datasets. The plots use a logarithmic scale for both axes, which prevents clustering of the points.}
    \label{fig:error-experiments-theory}
    \vspace{-12pt}
\end{figure}

\subsection{Analysis of \texorpdfstring{$\rhohf$}{ρf}}
\label{sec:6-rho-analysis}

As seen in Figure \ref{fig:error-experiments-theory}, the experimental error follows the overall shape described by the theoretical bound in Proposition \ref{prop:bound-error-rho}. According to the theory, the constants involved (e.g., the intercept of the curves) should depend on $\rhohf$. Hence, it is crucial to have estimates for this metric. In Appendix \ref{sec:appendix-5} we describe a procedure that can provide an estimate $\rhohat$ of $\rhohf$. The values computed with this method are (\textit{usparse}, 1.20), (\textit{normal}, 3.89), (\textit{amzn}, 1.72), (\textit{osm}, 32.57). The value for \textit{osm} is an order of magnitude greater than the rest, meaning it is farther away from the uniform distribution as measured by the Rényi divergence, helping to explain why it is challenging. The dataset with the lowest value is \textit{usparse}. Since its underlying distribution is uniform, this is in line with the analysis in Section \ref{sec:5-renyi}. Figure \ref{fig:error-experiments-theory} shows that higher values of $\rhohat$ correspond to higher experimental errors.

\section{Conclusions}
\label{sec:7}
Theoretical understanding of learned indexes has not kept up with their practical development. We narrow this gap by proving that learned indexes exhibit strong theoretical guarantees under a very general probabilistic model. In particular, we prove learned indexes can achieve $O(1)$ expected query time with at most linear space. We introduce a specific index that achieves this performance, and we find a general bound for the expected time in terms of the number of keys and the Rényi entropy of the distribution. From this result, we derive a metric $\rhohf$ that can be used to characterize the difficulty of learning the data. We describe a procedure for estimating $\rhohf$ for real data, and we show that the complexity bounds hold in practice. Our results help explain the good experimental performance of learned indexes, under the most general conditions considered so far in the literature. As future work, it is theoretically important to prove lower bounds for expected query time. We believe this type of result may hinge on information-theoretical properties \general{(e.g., through rate-distortion theory \cite{berger2003rate})}, which our analysis shows are relevant for learned indexes. On the other hand, it is also important to extend our analysis to the case of dynamic datasets, where the probability distribution may exhibit drift over time and updates are necessary. This is an important case in applications, and has not been considered in theoretical analysis.

\clearpage
\bibliography{main}

\clearpage
\appendix

\section{Proof of Lemma \ref{lemma:rank-error-bound-number-keys}}
\label{sec:appendix-1}
\begin{proof}[Proof of Lemma \ref{lemma:rank-error-bound-number-keys}]
    Since $q$ belongs to subinterval $I_k=[t_{k-1}, t_k]$, we know that $\rankh(q)=\hat{r}_k$. Also, since $\rank$ is an increasing function, it holds that $\rank(t_{k-1}) \leq \rank(q) \leq \rank(t_{k})$. Subtracting $\hat{r}_k$ from this inequality, we get
    \begin{equation}
    \label{eq:lemma-rank-error-bound-number-keys}
        \rank(t_{k-1}) - \hat{r}_k
        \leq \rank(q) - \hat{r}_k
        \leq \rank(t_{k}) - \hat{r}_k.
    \end{equation}
    As described in Section \ref{sec:3-espc-index}, $\hat{r}_k$ can be written as $\hat{r}_k = \rank(t_{k-1}) + n_k/2$. Hence, it follows that $\rank(t_{k-1}) - \hat{r}_k = -n_k/2$. On the other hand, it can be seen directly from the definition of $\rank$ that $\rank(t_{k})$ is equal to $\rank(t_{k-1}) + n_k$. This means that
    \begin{equation*}
        \rank(t_{k}) - \hat{r}_k
        = \rank(t_{k-1}) + n_k - \hat{r}_k
        = (\rank(t_{k-1}) - \hat{r}_k) + n_k
        = -\frac{n_k}{2} + n_k
        = \frac{n_k}{2}.
    \end{equation*}
    Substituting into equation (\ref{eq:lemma-rank-error-bound-number-keys}), we get the desired result:
    \begin{equation*}
        -\frac{n_k}{2}
        \leq \rank(q) - \hat{r}_k
        \leq \frac{n_k}{2}.
    \end{equation*}
\end{proof}

\section{Proofs of Main Complexity Bounds}
\label{sec:appendix-2}
\begin{proof}[Proof of Theorem \ref{thm:constant-time}]
    Let $n\in\mathbb{N}$ and consider the following procedure $R_n$ for generating an index: given a sorted array $A = [\osx{1}, \ldots, \osx{n}]$, the output $R_n(A)$ consists of an ESPC index built with $K=n$ subintervals. By Proposition \ref{prop:espc-basic-complexity}, this means that space overhead is $S(R_n(A)) = O(K) = O(n)$. As this is independent of the specific values stored in $A$, we have $\spacen = O(n)$. On the other hand, by Theorem \ref{thm:bound-error-K} and Proposition \ref{prop:bound-error-rho} we can bound the expected prediction error by
    \reviewerone{
        \begin{equation*}
            \expec{\eps}
            \leq
            \frac{3(b-a)}{2}\rhohf\frac{n}{n}
            =
            \frac{3(b-a)}{2}\rhohf.
        \end{equation*}
        Furthermore, Jensen's inequality implies
        \begin{equation*}
            \expec{\log \eps}
            \leq \log \expec{\eps}
            \leq \log \big((b-a)\rhohf\big) + \log (3/2)
            \leq \log \big((b-a)\rhohf\big) + 1.
        \end{equation*}
        By Proposition \ref{prop:espc-basic-complexity}, we know query time for $q$ for the ESPC index is $O(\log \eps(q))$. Accordingly, the expected query time is $\timen = O(\expec{\log \eps})=O\big(\log\big((b-a\big)\rhohf\big)+1\big)=O\big(\log\big((b-a\big)\rhohf\big)\big)$. This concludes the proof.
    }
\end{proof}

\begin{proof}[Proof of Theorem \ref{thm:loglog-time}]
    We use the same argument as the proof for Theorem \ref{thm:constant-time}, but setting $K=\frac{n}{\log n}$ subintervals instead of $K=n$.
\end{proof}

\section{Proof of Proposition \ref{prop:bound-error-rho-query}}
\label{sec:appendix-3}
\probanalysis{
    \begin{proof}[Proof of Proposition \ref{prop:bound-error-rho-query}]
        The proof follows closely the argument used to prove Theorem \ref{thm:bound-error-K} and Proposition \ref{prop:bound-error-rho}. By the tower property for conditional expectation,
        \begin{align*}
            \expec{\eps}
            = \expec{\expec{\eps \>|\> X_1, \ldots, X_n}}.
        \end{align*}
        The inner expected value gives
        \begin{equation*}
            \expec{\eps \>|\> X_1, \ldots, X_n}
            = 
            \int_\mathbb{R} \eps(q) \densityq(q) dq
            =
            \int_{\osx{1}}^{\osx{n}} \eps(q) \densityq(q) dq.
        \end{equation*}
        where the last equality comes from the fact that for $q\notin [\osx{1},\osx{n}]$ we know the exact value of $\rank(q)$, which is either $0$ or $n$, and hence we have that $\eps(q)=0$. Replacing $\eps(q)=|\rank(q)-\rankh(q)|$ and noticing that the $\{I_k\}$ form a partition of $[\osx{1}, \osx{n}]$:
        \begin{equation*}
            \expec{\eps \>|\> X_1, \ldots, X_n}
            =
            \sum_{k=1}^K\int_{I_k} |\rank(q)-\rankh(q)| \densityq(q) dq.
        \end{equation*}
        By Lemma \ref{lemma:rank-error-bound-number-keys} we can bound the error within each $I_k$:
        \begin{equation*}
            \expec{\eps \>|\> X_1, \ldots, X_n}
            \leq
            \frac{1}{2}\sum_{k=1}^K n_k \int_{I_k} \densityq(q) dq,
        \end{equation*}
        where $n_k$ is the number of keys in subinterval $I_k$. Now, denote by $n_{k,\ell}$ the number of keys in $I_k \cap J_\ell$. Then, we can write
        \begin{equation}
        \label{eq:conditional-error-in-terms-of-I-g}
            \expec{\eps \>|\> X_1, \ldots, X_n}
            \leq
            \frac{1}{2}\sum_{k=1}^K \left(\sum_{\ell = 1}^K n_{k, \ell}\right) \int_{I_k} \densityq(q) dq
            =
            \frac{1}{2}\sum_{\ell=1}^K \sum_{k = 1}^K n_{k, \ell} \int_{I_k} \densityq(q) dq,
        \end{equation}
        where first we have written $n_k$ as $\sum_{\ell = 1}^K n_{k, \ell}$ and then we have exchanged the order of summation. Now, denote by $L(k)$ the set of $\ell$ indices such that $I_k\cap J_\ell \neq \emptyset$. Notice that $n_{k, \ell}=0$ if $\ell\notin L(k)$. Hence, inequality (\ref{eq:conditional-error-in-terms-of-I-g}) becomes:
        \begin{equation}
        \label{eq:conditional-error-in-terms-of-I-and-J-g}
            \expec{\eps \>|\> X_1, \ldots, X_n}
            \leq
            \frac{1}{2}\sum_{\ell=1}^K \sum_{k: \ell\in L(k)} n_{k, \ell} \int_{I_k} \densityq(q) dq.
        \end{equation}
        By Fact \ref{lemma:intervals-subintervals} we know that $I_k\subseteq J_{\ell-1} \cup J_{\ell} \cup J_{\ell+1}$ for all $\ell$ and $k$ such that $\ell\in L(k)$. Now, define $q_\ell$ as $\mathbb{P}\big(q\in J_\ell\big)$ for each $\ell=1,\ldots, K$, analogously to how the $p_\ell$ are defined as $\mathbb{P}\big(X_i\in J_\ell\big)$. Using this notation, from inequality (\ref{eq:conditional-error-in-terms-of-I-and-J-g}) we get
        \begin{align}
            \expec{\eps \>|\> X_1, \ldots, X_n}
            &\leq
            \frac{1}{2}
            \sum_{\ell=1}^K \sum_{k: \ell\in L(k)} n_{k, \ell} \left(
                \int_{J_{\ell-1}} \densityq(q) dq
                + \int_{J_{\ell}} \densityq(q) dq
                + \int_{J_{\ell+1}} \densityq(q) dq
            \right) \nonumber \\
            &=
            \frac{1}{2} \sum_{\ell=1}^K (q_{\ell-1}+q_{\ell}+q_{\ell+1})
            \sum_{k: \ell\in L(k)} n_{k, \ell} \nonumber \\
            &=
            \frac{1}{2} \sum_{\ell=1}^K (q_{\ell-1}+q_{\ell}+q_{\ell+1}) m_\ell \label{eq:conditional-error-in-terms-of-J-g},
        \end{align}
        where $m_\ell$ denotes the number of keys in $J_\ell$. Now, applying expected value on both sides of inequality (\ref{eq:conditional-error-in-terms-of-J-g}), we get that
        \begin{equation}
        \label{eq:expected-error-in-terms-of-q-m}
            \expec{\eps}
            = \expec{\expec{\eps \>|\> X_1, \ldots, X_n}}
            \leq \frac{1}{2}\sum_{\ell=1}^K (q_{\ell-1}+q_{\ell}+q_{\ell+1}) \expec{m_\ell}
        \end{equation}
        Furthermore, we have
        \begin{equation*}
            \expec{m_\ell}
            = \expec{\sum_{i=1}^n \mathbf{1}_{X_i\in J_\ell}}
            = \sum_{i=1}^n \expec{\mathbf{1}_{X_i\in J_\ell}}
            = \sum_{i=1}^n \prob{X_i\in J_\ell}
            = n p_\ell.
        \end{equation*}
        Replacing $\expec{m_\ell}=n p_\ell$ in inequality (\ref{eq:expected-error-in-terms-of-q-m}), we get:
        \begin{equation}
        \label{eq:expected-error-in-terms-of-p-q}
            \expec{\eps}
            \leq
            \frac{1}{2}\sum_{\ell=1}^K (q_{\ell-1}+q_{\ell}+q_{\ell+1}) \expec{m_\ell}
            =
            \frac{n}{2}\sum_{\ell=1}^K (q_{\ell-1}p_{\ell}+q_{\ell}p_{\ell}+q_{\ell+1}p_{\ell}).
        \end{equation}
        Now, by Cauchy-Schwarz inequality we know that
        \begin{equation*}
            \sum_{\ell=1}^K q_{\ell-1}p_{\ell}
            \leq
            \left(
                \sum_{\ell=1}^K q_{\ell-1}^2\sum_{\ell=1}^K p_{\ell}^2
            \right)^{1/2}
            =
            \left(
                \sum_{\ell=0}^{K-1} q_{\ell}^2\sum_{\ell=1}^K p_{\ell}^2
            \right)^{1/2}
            \leq
            \left(
                \sum_{\ell=1}^K q_{\ell}^2\sum_{\ell=1}^K p_{\ell}^2
            \right)^{1/2},
        \end{equation*}
        where we have used the fact that $q_0=\prob{q\in J_0}=0$ since $J_0$ is defined as the empty set. Similarly, by Cauchy-Schwarz inequality and the fact that $q_{K+1}=0$ we can prove
        \begin{equation*}
            \sum_{\ell=1}^K q_{\ell}p_{\ell}
            \leq
            \left(
                \sum_{\ell=1}^K q_{\ell}^2\sum_{\ell=1}^K p_{\ell}^2
            \right)^{1/2},
            \>\>\>\>\>\>
            \sum_{\ell=1}^K q_{\ell+1}p_{\ell}
            \leq
            \left(
                \sum_{\ell=1}^K q_{\ell}^2\sum_{\ell=1}^K p_{\ell}^2
            \right)^{1/2}.
        \end{equation*}
        Replacing in (\ref{eq:expected-error-in-terms-of-p-q}) we get the new inequality
        \begin{equation*}
            \expec{\eps}
            \leq
            \frac{3n}{2}
            \left(\sum_{\ell=1}^K q_{\ell}^2\right)^{1/2}
            \left(\sum_{\ell=1}^K p_{\ell}^2\right)^{1/2}.
        \end{equation*}
        Now, since both $\densityx$ and $\densityq$ are square-integrable functions, the same argument from Proposition \ref{prop:bound-error-rho} can be used to prove that
        \begin{align*}
            \sum_{\ell=1}^K q_\ell^2
            \leq
            \frac{(b-a)}{K}\rhohg,
            \>\>\>\text{and}\>\>\>
            \sum_{\ell=1}^K p_\ell^2
            \leq
            \frac{(b-a)}{K}\rhohf.
        \end{align*}
        Putting everything together:
        \begin{align*}
            \expec{\eps}
            &\leq
            \frac{3n}{2}
            \left(\sum_{\ell=1}^K q_\ell^2\right)^{1/2}
            \left(\sum_{\ell=1}^K p_\ell^2\right)^{1/2} \\
            &\leq
            \frac{3n}{2}
            \left(\frac{(b-a)}{K}\rhohg\right)^{1/2}
            \left(\frac{(b-a)}{K}\rhohf\right)^{1/2} \\
            &=
            \frac{3(b-a)}{2}\sqrt{\rhohf \rhohg}\frac{n}{K}.
        \end{align*}
        This concludes the proof.
    \end{proof}
}

\reviewerone{
    \section{Proof of Constant Time Results for Cases of Unbounded Support}
    \label{sec:appendix-4}
    \begin{proof}[Proof of Theorem \ref{thm:constant-time-unbounded-1}]
    Set $a=\mu - \sigma\sqrt{n\log n}$ and $b=\mu + \sigma\sqrt{n\log n}$. It holds that
    \begin{align*}
        \prob{A_{a, b}}
        &=
        \prob{[\osx{1}, \osx{n}]\not\subseteq [a,b]} \\
        &=
        \prob{\bigcup_{i=1}^n \big\{X_i\notin [a,b]\big\}}
    \intertext{By the union bound:}
        &\leq
        \sum_{i=1}^n \prob{X_i\notin [a,b]}
    \intertext{By definition of $a$ and $b$:}
        &=
        \sum_{i=1}^n \prob{
            \Big\{X_i < \mu - \sigma\sqrt{n\log n}\Big\}
            \cup
            \Big\{X_i > \mu + \sigma\sqrt{n\log n}\Big\}
        } \\
        &=
        \sum_{i=1}^n \prob{|X_i - \mu| > \sigma\sqrt{n\log n}}
    \intertext{By Chebyshev's inequality:}
        &\leq
        \sum_{i=1}^n \frac{1}{n\log n}
        =
        \frac{1}{\log n}.
    \end{align*}
    Since $\prob{A_{a, b}}\leq \frac{1}{\log n}$ equation (\ref{eq:expected-log-error-rho-n-K}) applies, and replacing $a$ and $b$ gives
    \begin{equation*}
        O(\expec{\log \eps})
        =
        O\left(\log \left((b-a)\rhohf\frac{n}{K}\right)\right)
        =
        O\left(\log \left(2\sigma\sqrt{n\log n}\rhohf\frac{n}{K}\right)\right).
    \end{equation*}
    Hence, building an ESPC index with $K=n\sqrt{n\log n}$ intervals makes space usage $O\Big(n\sqrt{n\log n}\Big)$ and guarantees that expected query time is $O\left(\log ( 2 \sigma \rhohf ) \right)$. This concludes the proof.
\end{proof}

\begin{proof}[Proof of Theorem \ref{thm:constant-time-unbounded-2}]
    The distribution of the $\{X_i\}$ is subexponential, so there is a constant $C>0$ such that $\prob{|X_i| \geq x} \leq 2e^{-Cx}$ for all $x$. Set $a=-\frac{2\log n}{C}$, $b=\frac{2\log n}{C}$. It holds that
    \begin{align*}
        \prob{A_{a, b}}
        &=
        \prob{[\osx{1}, \osx{n}]\not\subseteq [a,b]} \\
        &=
        \prob{\bigcup_{i=1}^n \big\{X_i\notin [a,b]\big\}}
    \intertext{By the union bound:}
        &\leq
        \sum_{i=1}^n \prob{X_i\notin [a,b]}
    \intertext{By definition of $a$ and $b$:}
        &=
        \sum_{i=1}^n \prob{
            \left\{X_i < -\frac{2\log n}{C}\right\}
            \cup
            \left\{X_i > \frac{2\log n}{C}\right\}
        } \\
        &=
        \sum_{i=1}^n \prob{|X_i| > \frac{2\log n}{C}}
    \intertext{Because the distribution is subexponential:}
        &\leq
        \sum_{i=1}^n 2e^{-C\frac{2\log n}{C}} \\
        &=
        \frac{2}{n}
        \leq
        \frac{1}{\log n},
    \end{align*}
    for all $n\geq 2$. Since $\prob{A_{a, b}}\leq \frac{1}{\log n}$ equation (\ref{eq:expected-log-error-rho-n-K}) applies, and replacing $a$ and $b$ gives
    \begin{equation*}
        O(\expec{\log \eps})
        =
        O\left(\log \left((b-a)\rhohf\frac{n}{K}\right)\right)
        =
        O\left(\log \left(\frac{4\log n}{C}\rhohf\frac{n}{K}\right)\right).
    \end{equation*}
    Hence, building an ESPC index with $K=n \log n$ subintervals makes space usage $O(n\log n)$ and guarantees that expected query time is $O\left(\log \big( 4 \rhohf / C \big) \right)$. This concludes the proof.
\end{proof}

}

\section{Estimation of \texorpdfstring{$\rhohf$}{ρf} metric}
\label{sec:appendix-5}
As the key measure of statistical complexity for the data generating process, $\rhohf$ can help explain the experimental differences in the performance of the ESPC index for different datasets. In that sense, it is important to have an estimate for this metric. Notice that
\begin{equation*}
    \rhohf
    = \int_\mathbb{R} \densityx^2(x)dx
    = \int_\mathbb{R} \densityx(x) \densityx(x)dx
    = \expec{\densityx},
\end{equation*}
that is, $\rhohf$ is equal to the expected value of the density function. In other words, if we can
\begin{enumerate}
    \item\label{step:monte-carlo-1} Sample $J$ values $z_1,\ldots,z_J$ from the probability distribution defined by $\densityx$, and
    \item\label{step:monte-carlo-2} Evaluate $\densityx(z_j)$ for all $j\in\{1,\ldots,J\}$,
\end{enumerate}
then we can get an unbiased estimator for $\rhohf$ via
\begin{equation*}
    \rhohat
    =
    \frac{1}{J}\sum_{j=1}^J \densityx(z_j).
\end{equation*}
If the $\{z_j\}$ are sampled independently, this corresponds to a Monte Carlo estimator and we know its variance goes to $0$ as $J\rightarrow\infty$.

We can adapt this procedure for real datasets. On the one hand, Step (\ref{step:monte-carlo-1}) can be simulated by sampling keys at random from array $A$. On the other hand, since $\densityx$ is usually not known, for Step (\ref{step:monte-carlo-2}) we use an estimate $\hat{\densityx}$ of the density $\densityx$. This estimate can be derived via several methods (e.g., kernel density estimation \cite{scott2015multivariate}) with different properties in terms of mean squared or uniform error.

We mostly use the histogram method, which under mild assumptions can be shown to have vanishing bias as the bin width goes to $0$ \cite{freedman1981histogram}. For this method, the mean squared error is minimized with $\Omega(n^{-1/3})$ bin width, where $n$ is the number of keys. One robust option in practice is to choose the bin width using the Freedman–Diaconis rule \cite{freedman1981histogram}. For most datasets, generated from either synthetic (\textit{usparse}, \textit{normal}) or real-world data (\textit{amzn}), this method exhibits the convergence expected from Monte Carlo estimation as the number of samples increases. Figure \ref{fig:rho-estimate-monte-carlo} exemplifies this with the \textit{normal} and \textit{amzn} datasets.

\begin{figure}[ht]
    \centering
    \includegraphics[width=0.8\linewidth]{
        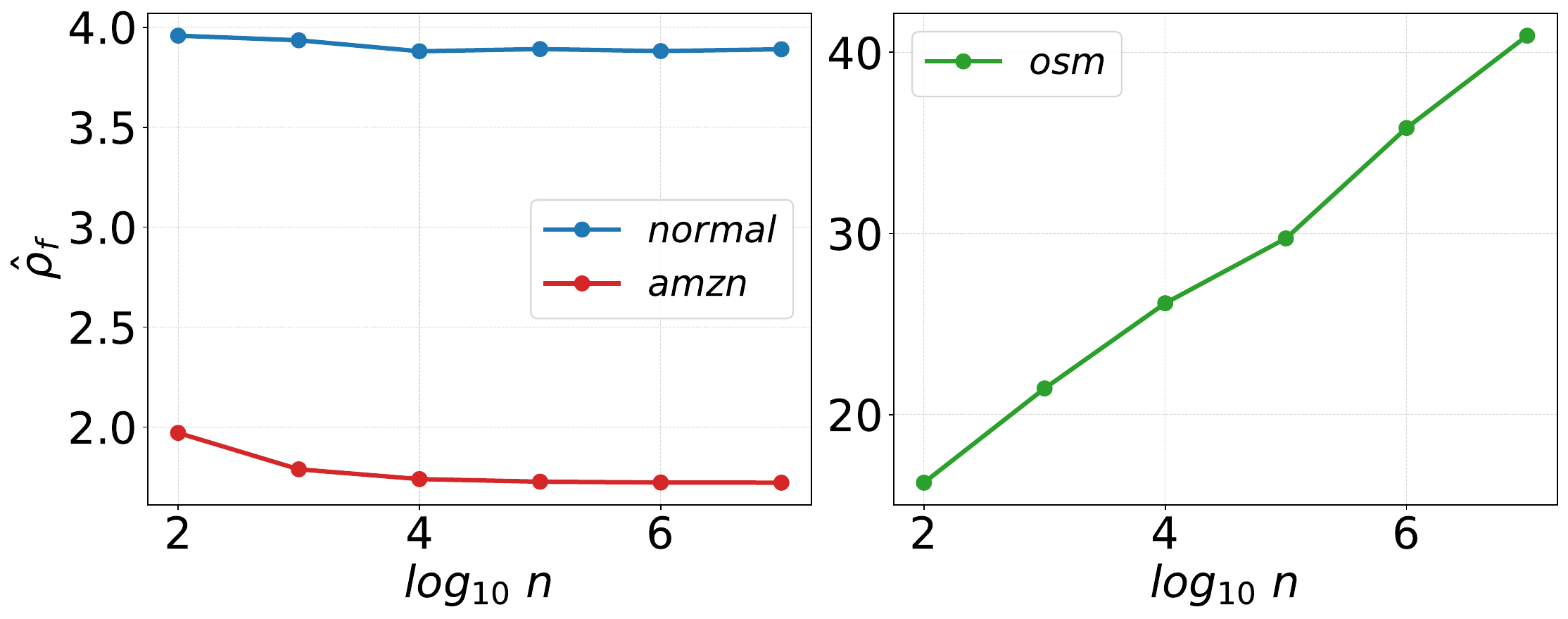
    }
    \caption{$\rhohat$ estimate for \textit{normal}, \textit{amzn} and \textit{osm} datasets as number of samples $n$ increases. The estimation process uses the histogram method to approximate the density $\densityx$.}
    \label{fig:rho-estimate-monte-carlo}
\end{figure}

For the \textit{osm} dataset, this convergent behavior is not observed, as Figure \ref{fig:rho-estimate-monte-carlo} shows. It may be that the variance of the estimator is higher and more samples are needed to get an accurate estimate. In light of this, we calculate $\rhohat$ by using a kernel density estimator for the density $\densityx$, which seems to be more robust than the histogram estimator for this dataset.

Table \ref{tab:rho-estimates} shows the estimates $\rhohat$ obtained through the use of our method. As can be seen, most datasets exhibit low values for this metric, consistent with the good performance that learned indexes have in general. Among these datasets, \textit{osm} has the highest value of $\rhohat$, consistent with the fact that it is considered a challenging benchmark for learned indexes.

\begin{table}[ht]
    \centering
    \begin{tabular}{|r r|} 
         \hline
         \textbf{Dataset} & $\rhohat$ \\
         \hline
            \textit{usparse}   & $1.20$             \\
            \textit{normal}    & $3.89$             \\
            \textit{amzn}      & $1.72$             \\
            \textit{osm}       & $32.57$            \\ [1ex]
         \hline
    \end{tabular}
    \caption{Estimate $\rhohat$ for each dataset used in our experiments. Most have low values, consistent with good performance of learned indexes. The \textit{osm} dataset exhibits a high value of $\rhohat$, an order of magnitude greater than the rest, helping to explain why it is considered challenging.}
    \label{tab:rho-estimates}
\end{table}

\end{document}